\documentclass[aop,preprint]{imsart}

\usepackage{amsmath}
\usepackage{amssymb}
\usepackage{amsfonts}
\usepackage{latexsym}
\usepackage{amsthm}
\usepackage{amsxtra}
\usepackage{amscd}
\usepackage{mathrsfs}
\usepackage{bm}

\usepackage{graphicx, color}
\usepackage[nottoc,notlof,notlot]{tocbibind} 
\usepackage{cite}

\arxiv{arXiv:1802.02077}

\startlocaldefs

\theoremstyle{plain} %plain, definition, remark
\newtheorem{theorem}{Theorem}[section]
\newtheorem*{theorem*}{Theorem}
\newtheorem{lemma}[theorem]{Lemma}
\newtheorem*{lemma*}{Lemma}
\newtheorem{corollary}[theorem]{Corollary}
\newtheorem*{corollary*}{Corollary}

\newtheorem*{proposition*}{Proposition}
\newtheorem{definition}[theorem]{Definition}
\newtheorem*{definition*}{Definition}

\newtheorem*{conjecture*}{Conjecture}

\newtheorem*{example*}{Example}
\newtheorem{remark}[theorem]{Remark}
\newtheorem*{remark*}{Remark}

\definecolor{darkred}{rgb}{0.9,0,0.3}
\definecolor{darkblue}{rgb}{0,0.3,0.9}

\usepackage{ifthen}
\def\comment#1{\ifthenelse{\isodd{\value{page}}}{\marginpar{\raggedright\scriptsize{\textcolor{darkred}{#1}}}}{\marginpar{\raggedleft\scriptsize{\textcolor{darkred}{#1}}}}}

\newcommand{\E}{\mathbb{E}}
\newcommand{\R}{\mathbb{R}}
\newcommand{\C}{\mathbb{C}}

\newcommand{\N}{\mathbb{N}}
\newcommand{\Z}{\mathbb{Z}}

\newcommand{\HH}{\mathbb{H}}

\renewcommand{\leq}{\leqslant}
\renewcommand{\geq}{\geqslant}
\renewcommand{\epsilon}{\varepsilon}

\newcommand{\pbb}[1]{\biggl({#1}\biggr)}

\newcommand{\pa}[1]{\left({#1}\right)}

\newcommand{\avg}[1]{\langle #1 \rangle}

\newcommand{\avga}[1]{\left\langle #1 \right\rangle}

\DeclareMathOperator{\sdet}{sdet}

\newcommand{\Hhoro}{\widetilde H}
\newcommand{\Fhoro}{\widetilde F}
\newcommand{\Omegahoro}{\widetilde \Omega}

\newcommand{\psibar}{\bar{\psi}}

\newcommand{\ddp}[2]{\frac{\partial #1}{\partial #2}}

\newcommand{\bydef}{\equiv}

\endlocaldefs

\begin{document}

\begin{frontmatter}

% "Title of the paper"
\title{Dynkin isomorphism and Mermin--Wagner theorems for
  hyperbolic sigma models and recurrence of the two-dimensional
  vertex-reinforced jump process}
\runtitle{Hyperbolic sigma models and VRJP}

% indicate corresponding author with \corref{}
% \author{\fnms{John} \snm{Smith}\corref{}\ead[label=e1]{smith@foo.com}\thanksref{t1}}
% \thankstext{t1}{Thanks to somebody} 
% \address{line 1\\ line 2\\ printead{e1}}
% \affiliation{Some University}

\begin{aug}
\author{\fnms{Roland} \snm{Bauerschmidt}\corref{}\ead[label=e1]{rb812@cam.ac.uk}},
\author{\fnms{Tyler} \snm{Helmuth}\ead[label=e2]{th17948@bristol.ac.uk}}
%\address{\printead{e2}}
\and
\author{\fnms{Andrew} \snm{Swan}\ead[label=e3]{acks2@cam.ac.uk}}
%\address{\printead{e3}}

\runauthor{R.\ Bauerschmidt et al.}

\affiliation{University of Cambridge, Statistical Laboratory, DPMMS}
\and
\affiliation{University of Bristol, School of Mathematics}

\address{University of Cambridge, \\
Centre for Mathematical Sciences \\
Wilberforce Road \\
Cambridge, CB3 0WB, UK \\ 
\printead{e1,e3}}
\address{University of Bristol, \\
School of Mathematics\\
University Walk\\
Bristol, BS8 1TW, UK \\
\printead{e2}}
\end{aug}

\begin{abstract}
  We prove the vertex-reinforced jump process (VRJP) is recurrent in
  two dimensions for any translation invariant finite-range
  initial rates. Our proof has two main ingredients. The first is a
  direct connection between the VRJP and sigma models whose
  target space is a hyperbolic space $\HH^{n}$ or its supersymmetric
  counterpart $\HH^{2|2}$. These results are analogues of well-known
  relations between the Gaussian free field and the local times of
  simple random walk. The second ingredient is a Mermin--Wagner
  theorem for these sigma models. This result is of intrinsic
  interest for the sigma models and also implies our main
  theorem on the VRJP. Surprisingly, our Mermin--Wagner theorem
  applies even though the symmetry groups of $\HH^n$ and $\HH^{2|2}$
  are non-amenable.
\end{abstract}



\end{frontmatter}

\section{Introduction and results}

\subsection{Introduction}
\label{sec:intro}

Our results have motivation from two different perspectives, that of
sigma models with hyperbolic symmetry and their relevance for the
Anderson transition, and that of a model of reinforced
random walks known as the \emph{vertex-reinforced jump process} (VRJP).

The VRJP was originally introduced by Werner and has attracted a great
deal of attention
recently~\cite{MR1900324,MR3420510,SabotZengPubl,MR3663098,MR3366053}.
The VRJP on a vertex set $\Lambda$ is a continuous-time random walk
that jumps from a vertex $i$ to a neighbouring vertex $j$ at time $t$
with rate 
$\beta_{ij}(1+L_{t}^{j})$, where $L_{t}^{j}$ is the local time of $j$
at time $t$ and $\beta_{ij}\geq 0$ are the initial rates.  One
should view $\Lambda$ as the vertex set of an undirected graph with edge set
$E = \{ \avg{ij} \mid \beta_{ij}>0\}$.  The dependence of
the jump rates on the local time leads the VRJP to be
attracted to itself.

One of our new results is the following theorem.
\begin{theorem} \label{thm:recurrence} Consider a vertex-reinforced
  jump process $(X_t)$ on the vertex set $\Z^{d}$ with initial rates
  $\beta$ that are finite-range and translation invariant. If $d=1,2$
  then $(X_t)$ is recurrent in the sense that the expected time
  $(X_t)$ spends at the origin is infinite.
\end{theorem} 

As the VRJP is not a Markov process, different notions of recurrence
are not \emph{a priori} equivalent. For example, another natural
notion of recurrence would be to ask if the VRJP visits the origin
infinitely often almost surely. For non-Markovian processes neither of
these definitions of recurrence implies the other: there may be
infinitely many visits to the origin with the increments of the local
time being summable. To the best of our knowledge, neither implication
is known for the VRJP.

For sufficiently small initial rates recurrence results for the VRJP
have previously been established \cite{MR3189433,MR3420510,MR2736958}.
These results are for recurrence in the sense of visiting the origin
infinitely often almost surely. See \cite{MR3189433} for a discussion
and precise statements.  It has also been shown that the linearly
edge-reinforced random walk (ERRW) with constant initial weights is
recurrent in two dimensions~\cite{MR2561431,SabotZengPubl}, but the
recurrence of the VRJP for all initial rates was an open problem until
the present work.  The relation between the ERRW and VRJP is discussed
below.

Theorem~\ref{thm:recurrence} is in fact a consequence of our proof of
a Mermin--Wagner theorem for hyperbolic sigma models and a new and very
direct relation between VRJPs and hyperbolic sigma models that
parallels the well-known relationship between simple random walks
and Gaussian free fields (the BFS--Dynkin isomorphism theorem).

Before giving precise definitions of our models and stating our
results, we briefly indicate the motivations behind hyperbolic sigma
models, and their relations with reinforced random walks. We also
explain some consequences of our results for hyperbolic sigma
models. Readers primarily interested in the VRJP may wish to skip
ahead to Section~\ref{sec:model-definitions}.

Hyperbolic sigma models were introduced as effective models to
understand the Anderson transition
\cite{MR1134935,MR2728731,MR2104878,MR3204347,MR2953867}.  In Efetov's
supersymmetric method~\cite{MR708812} the expected absolute value
squared of the resolvent of random band matrices,
i.e., $\E|(H-z)^{-1}(i,j)|^2$ where $z \in \C_+$ and $H$ is a random band
matrix, can be expressed as a correlation function of a
supersymmetric spin model. The spins of this model are invariant under
the hyperbolic symmetry $OSp(2,1|2)$.  Extended states correspond to
spontaneous breaking of this non-compact symmetry.  The supersymmetric
hyperbolic sigma model, or $\HH^{2|2}$ model, was introduced by
Zirnbauer~\cite{MR1134935} and first studied by Disertori, Spencer and
Zirnbauer \cite{MR2728731}. It is an approximation of the random band
matrix model above where radial fluctuations are neglected. This is
similar to how the $O(n)$ model is an approximation of models of
$\R^{n}$-valued spins with rotational symmetry such as
$|\varphi|^{4}$-theories.  More detailed motivation for hyperbolic
spin models is given in~\cite{MR2104878,MR2953867}.

The $\HH^{2|2}$ model is believed to capture the physics of the
Anderson transition.  As is expected for the Anderson model, it was
proved in~\cite{MR2728731} that the $OSp(2,1|2)$ symmetry of the
$\HH^{2|2}$ model is spontaneously broken in $d\geq 3$ for
sufficiently small disorder --- consistent with the existence of
extended states.  Furthermore, it was proved \cite{MR2736958} that for
sufficiently large disorder this is not the case --- consistent with
Anderson localisation.  In dimension $d\leq 2$, it is conjectured that
extended states do not exist for any disorder strength.
Equation~\eqref{e:GasympLE2} below is the corresponding statement for
the $\HH^{2|2}$ model, and we have thus completed the expected
qualitative picture for the phase diagram of the $\HH^{2|2}$ model;
see Remark~\ref{rk:localisation} for a discussion of the conjectured
optimal bounds. Equation~\eqref{e:GasympLE2} can be considered as a
version of the Mermin--Wagner theorem.  For recent and extremely
precise results in dimension one, see \cite{MR3824956}.

Based on the similarity of certain explicit formulas, it was suggested
that there is a connection between the $\HH^{2|2}$ model and linearly
edge-reinforced random walks~\cite{MR2728731}.  This connection was
first confirmed in~\cite{MR3420510} by relating marginals of the
$\HH^{2|2}$ model to the limiting local time profile of a time
  change of the VRJP.  It was also shown there that the linearly edge
reinforced walk is obtained from the VRJP when averaging over random
initial rates.  Further marginals of the $\HH^{2|2}$ model were
explored in~\cite{MR3663098}.  For a discussion of the history of the
VRJP, see \cite{MR3420510}.

Our hyperbolic analogue of the BFS--Dynkin isomorphism theorem,
Theorem~\ref{thm:VRJP-Dynkin} below, is a different relation between
the $\HH^{2|2}$ model and the VRJP than was found in \cite{MR3420510},
and it provides a more direct relation between the correlation
structures of the models.  Moreover, our statement also applies
without supersymmetry, i.e., when the spins take values in $\HH^{n}$.
We will explain further extensions of Theorem~\ref{thm:VRJP-Dynkin} in
the case of $\HH^{n}$, e.g., to multipoint correlations, in a
forthcoming publication.

\subsection{Model definitions}
\label{sec:model-definitions}

We now define the VRJP and the hyperbolic sigma models. The walk and
the sigma models are both defined in terms of a set $\Lambda$ of
vertices and non-negative edge weights $\beta=(\beta_{ij})_{i,j\in\Lambda}$,
where by edge weights we mean that $\beta_{ij}=\beta_{ji}$.
For our Mermin--Wagner theorem we will
make use of two assumptions on $\beta$. We call $\beta$ \emph{finite-range} if for each $i\in\Lambda$ we have $\beta_{ij}=0$ for all but
finitely many $j$. If $\Lambda =\Z^{d}$ we call $\beta$
\emph{translation invariant} if $\beta_{ij}=\beta_{T(i)T(j)}$ for all
translations $T$ of $\Z^{d}$.

\subsubsection{Vertex-reinforced jump process}
\label{sec:vrjp}

Let $\Lambda$ be a finite or countable set.  The VRJP is a
history-dependent continuous-time random walk $(X_t)$ on
$\Lambda$ that takes jumps from vertex $i$ to vertex $j$ with rate
$\beta_{ij}(1+L_{t}^{j})$, where
\begin{equation}
  L_t^j \bydef \int_0^t 1_{X_s=j} \, ds. 
\end{equation}
$L_{t}^{j}$ is called the \emph{local time} of the walk at vertex $j$ up
to time $t$. We will write $L_{t} \bydef (L_{t}^{i})_{i\in\Lambda}$
for the collection of local times.  It will also be useful to consider the
joint process $(X_{t},L_{t})$, which is a Markov process with
generator $\mathcal{L}$ acting on sufficiently nice functions
$g\colon \Lambda\times\R^{\Lambda}\to \R$ by
\begin{equation} \label{e:Ldef}
  \mathcal{L}^\beta g(i,\ell) = \sum_j
  \beta_{ij}(1+\ell_j) (g(j,\ell)-g(i,\ell)) + \ddp{}{\ell_i}
  g(i,\ell), \quad i\in\Lambda, \,\,\,\ell\in\R^{\Lambda}.
\end{equation}
We denote by $\E^\beta_{i,\ell}$ the expectation of the process
$(X_t,L_t)$ with initial condition $X_0=i$ and $L_0=\ell$.  The VRJP
is the marginal of $X_{t}$ in the special case $L_{0}=0$; by a slight
abuse of terminology we call $(X_{t},L_{t})$ the VRJP as well.

\subsubsection{Hyperbolic sigma models}
\label{sec:hyp}

Let $\R^{n,1}$ denote $(n+1)$-dimensional Minkowski space. Its
elements are vectors $u=(x,y^{1},\dots, y^{{n-1}},z)$, and it is
equipped with the indefinite inner product
$u\cdot u = x^2+(y^1)^2+\cdots+(y^{n-1})^2-z^2$.
Note that although $x$ plays the same role as the $y^i$, we distinguish it
in our notation for later convenience.
Recall that $n$-dimensional hyperbolic space $\HH^{n}$ can be realized as
\begin{equation} \label{e:Hn-def}
  \HH^{n} \bydef \{u\in \R^{n,1}\mid u\cdot u = -1, z>0\}.
\end{equation}

Suppose $\Lambda$ is finite and $h>0$.
To each vertex $i\in\Lambda$ we associate a
spin $u_i\in \HH^{n}$.  The energy of a spin configuration
$u=(u_i)_{i\in\Lambda} \in (\HH^{n})^\Lambda$ is
\begin{equation}
  \label{e:HH-def}
  H(u)
  = H_{\beta,h}(u)
  \bydef
  \sum_{ \avg{ij}} \beta_{ij}(-u_i\cdot u_j-1) + h \sum_j (z_j-1),
\end{equation}
where the sum is over edges $\avg{ij}$; since the summands are
symmetric in $i$ and $j$ this notation will not cause any confusion.
The $\HH^{n}$ sigma model is the measure with density
proportional to $e^{-H(u)}$ with respect to the $|\Lambda|$-fold
product of the measure $\mu$ on $\HH^{n}$ induced by the Minkowski
metric, see \eqref{e:HS-vol} and \eqref{e:Hhoro-Hn} for
explicit expressions, and we let $\avg{\cdot}_{\HH^{n}}$ denote the
expectation associated to this model:
\begin{equation}
  \label{e:HH-exp}
  \avg{F(u)}_{\HH^{n}} \bydef \frac{
    \int_{(\HH^n)^{\Lambda}}F(u) \, e^{-H(u)} \, \mu^{\otimes \Lambda}(du)}
  {\int_{(\HH^n)^{\Lambda}}e^{-H(u)} \, \mu^{\otimes \Lambda}(du)}.
\end{equation}
The energy \eqref{e:HH-def} favours spin alignment because
$u \cdot v \leq -1$ for $u,v\in \HH^{n}$ with equality if and only if $u=v$.

\subsubsection{Supersymmetric hyperbolic sigma model}
\label{sec:susy-hyp}

In this section we will introduce a probability measure which
enables the computation of a special class of
observables of the full supersymmetric $\HH^{2|2}$ model. These
restricted observables will suffice for a description of
a special, but interesting, case of our results.
Our most general results use the full supersymmetric formalism.

As will be explained further in Section~\ref{sec:susy},
at each vertex $i\in\Lambda$ there is a superspin
$u_i = (x_i,y_i,z_i,\xi_i,\eta_i) \in \HH^{2|2}$ where $\xi_i$ and
$\eta_i$ are Grassmann variables. 
For the moment all that is needed is that
the expectation of a function $F(y)$ of the
$y\bydef (y_i)_{i\in\Lambda}$ coordinates can be written as
\begin{equation}
  \label{e:H22-Real}
  \avg{F(y)}_{\HH^{2|2}} = \int_{(\R^2)^\Lambda} F(e^t s) \,
  e^{-\Hhoro(s,t)} \,  dt\,ds, 
\end{equation}
where $dt\,ds\bydef\prod_{i}dt_{i}\,ds_{i}$, $e^{t}s\bydef (e^{t_{i}}s_{i})_{i\in\Lambda}$, 
\begin{multline}
  \Hhoro(s,t)
  = \Hhoro_{\beta,h}(s,t)
  \bydef \sum_{ \avg{ij}} \beta_{ij} \pa{\cosh(t_i-t_j)-1+\frac12 (s_i-s_j)^2 e^{t_i+t_j}}
  \\ + h \sum_i \pa{\cosh(t_i)-1+\frac12 s_i^2 e^{t_i}} + \sum_i (t_i+\log(2\pi)) - \log \det D_{\beta,h}(t),
\end{multline}
and the matrix $D_{\beta,h}(t)$ on $\R^\Lambda$ is defined by the quadratic form
\begin{align}
  (v,D_{\beta,h}(t) v) \bydef \sum_{ \avg{ij}} \beta_{ij} e^{t_i+t_j} (v_i-v_j)^2 + h \sum_{i} e^{t_i} v_i^2, \quad v\in \R^\Lambda.
\end{align}
The determinant $\det D_{\beta,h}(t)$ does not depend on the $s$ variables and it is positive
since $D_{\beta,h}(t)$ is positive definite.  Thus
$e^{-\Hhoro(s,t)}dt\,ds$ is a positive measure, and 
we will show in Section~\ref{sec:susy} that it is in fact a
probability measure, i.e., $\avg{1}_{\HH^{2|2}}=1$. 

\subsection{Results}
\label{sec:results}

We now state our main results and show how
Theorem~\ref{thm:recurrence} is a consequence. 

\subsubsection{Hyperbolic BFS--Dynkin Isomorphism}
\label{sec:BFSD}

The following theorem is a hyperbolic analogue of the Dynkin isomorphism theorem,
which relates the local times of a simple random walk to the square of
a Gaussian free field.  As the Dynkin isomorphism theorem was proved
by Brydges--Fr\"ohlich--Spencer in \cite[Theorem~2.2]{MR648362}, and
later expressed in a better form by Dynkin \cite{MR693227}, we prefer
to call it the BFS--Dynkin isomorphism. The general idea of relating
Gaussian fields to simple random walks is due to
Symanzik~\cite{Syma69}. For recent discussions of these ideas
see~\cite{MR2932978, MR2815763}.  Supersymmetric versions of these
results for simple random walks go back to Luttinger and Le
Jan~\cite{MR941982,MR713539}.

Note that while we have not yet defined the meaning of
$\avg{g}_{\HH^{2|2}}$ for a general function $g$, we have given a
meaning in the case that $g$ is identically one by
\eqref{e:H22-Real}. It is this case of $g$ identically one that
  will be most relevant for the VRJP.

\begin{theorem} \label{thm:VRJP-Dynkin}
  Suppose $\Lambda$ is finite and $\beta$ is a collection of non-negative edge weights.
  Let $h>0$, let $g\colon \Lambda\times\R^{\Lambda}\to \R$ be any
  bounded smooth function, and let $a,b \in \Lambda$.
  Consider the $\HH^{n}$ model, $n\geq 2$,
  let $y = (y_i)_{i\in\Lambda}=(y^{r}_{i})_{i\in\Lambda}$ for some
  $r=1,\dots,n-1$,  and $z=(z_{i})_{i\in\Lambda}$. Then
  \begin{equation} \label{e:Eg-Dynkin-VRJP-thm}
    \sum_b  \avg{y_a y_b g(b,z-1)}_{\HH^n} 
    = 
    \avg{z_a \int_0^\infty
    \E^{\beta}_{a,z-1}(g(X_t,L_t)) \, e^{-ht} \, dt}_{\HH^{n}} .
  \end{equation}
  For the $\HH^{2|2}$ model, we have
  \begin{equation} \label{e:Eg-VRJP-thm}
    \sum_b  \avg{y_a y_b g(b,z-1)}_{\HH^{2|2}} 
    = \int_0^\infty \E^{\beta}_{a,0}(g(X_t,L_t)) \, e^{-ht} \, dt.
  \end{equation}
\end{theorem}
\begin{remark}
  Theorem~\ref{thm:VRJP-Dynkin} also holds for the $\HH^{1}$ model, but as the
  proof requires slightly different considerations we have not
  included it here.
\end{remark}

Taking the function $g$ to be identically one in \eqref{e:Eg-VRJP-thm} implies that
\begin{equation} \label{e:twopoint-H22-VRJP}
  \avg{y_ay_b}_{\HH^{2|2}} = \int_0^\infty \E_{a,0}^\beta(1_{X_t=b}) \, e^{-ht} \, dt.
\end{equation}
The right-hand side can be interpreted as the two-point function of
the VRJP with a uniform killing rate $h$.
\begin{remark}
  Theorem~\ref{thm:VRJP-Dynkin} can be extended in a straightforward
  way to the case in which $h = (h_{i})_{i\in\Lambda}$ is
  non-constant, provided $h_{i}\geq 0$ and at least one value is
  strictly positive.
\end{remark}

\subsubsection{Hyperbolic Mermin--Wagner Theorem}
\label{sec:HMW}

In this section we assume that $\Lambda = \Lambda_L$ is
the discrete $d$-dimensional torus $\Z^{d}/(L\Z)^{d}$ of side length $L\in\N$,
and that $\beta$ is translation invariant and finite-range.  We will
write $\avg{\cdot}=\avg{\cdot}_{\beta,h}$ in place of
$\avg{\cdot}_{\HH^{n}}$ and $\avg{\cdot}_{\HH^{2|2}}$.  Denote
\begin{equation}
  \label{e:lambda}
  \lambda(p) \bydef \sum_{j\in \Lambda} \beta_{0j} (1-\cos(p\cdot j)),
  \qquad p\in \Lambda^{\star},
\end{equation}
where here $\cdot$ is the Euclidean inner product on $\R^d$ and
$\Lambda^{\star}$ is the Fourier dual of the discrete torus $\Lambda$.
Denote the two-point function and its Fourier transform by
\begin{equation}
  G_{\beta,h}(j) =G^{L}_{\beta,h}(j) \bydef \avg{y_0y_j}_{\beta,h}, \qquad
  \hat G_{\beta,h}(p) = \hat G^{L}_{\beta,h}(p) =\sum_{j\in\Lambda} G_{\beta,h}(j) e^{i(p\cdot j)}.
\end{equation}
The following theorem is an analogue of the Mermin--Wagner Theorem for
the $O(n)$ model, in the form presented in \cite{MR610687}.

\begin{theorem} \label{thm:MW}
Let $\Lambda = \Z^{d}/(L\Z)^{d}$, $L\in\N$. For the
$\HH^{n}$ model, $n\geq 2$, with magnetic field $h > 0$,
\begin{equation} \label{e:MW-Hn}
  \hat G_{\beta,h}(p)
  \geq
  \frac{1}{ (1+(n+1)G_{\beta,h}(0)) \lambda(p) + h}
  .
\end{equation}
Similarly, for the $\HH^{2|2}$ model with $h > 0$, 
\begin{equation} \label{e:MW-H22}
  \hat G_{\beta,h}(p)
  \geq
  \frac{1}{ (1+G_{\beta,h}(0)) \lambda(p) + h}
  .
\end{equation}
\end{theorem}
\begin{remark}
  \label{rem:twopoint-H22-VRJP}
  By \eqref{e:twopoint-H22-VRJP} the two-point function $G_{\beta,h}$
  equals that of the VRJP in the case of the $\HH^{2|2}$ model, and
  hence the two-point function of the VRJP satisfies \eqref{e:MW-H22}
  as well.
\end{remark}

\begin{remark}
  For $d \geq 3$, the bound \eqref{e:MW-H22}  shows that $\tilde
  f$ can be replaced by $f$ in \cite[Theorem~3]{MR2728731} using the upper bound
  proved there for $G_{\beta,h}(0)$.
\end{remark}

\begin{corollary}
  \label{cor:GasympLE2}
  Under the assumptions of Theorem~\ref{thm:MW},
  for $d=1,2$,
  \begin{equation} \label{e:GasympLE2}
    \lim_{h\downarrow 0} \lim_{L \to\infty} G_{\beta,h}(0) = \infty.
  \end{equation}
\end{corollary}

\begin{proof}
Since $(2\pi L)^{-d} \sum_{p \in \Lambda^*} e^{i(p\cdot j)} = 1_{j=0}$,
summing the bounds \eqref{e:MW-Hn} and \eqref{e:MW-H22} over
$p\in\Lambda^{\star}$ and interchanging sums implies (with $n=0$ for $\HH^{2|2}$)
\begin{equation}
  G_{\beta,h}(0)
  \geq
  \frac{1}{(2\pi L)^d} \sum_{p\in\Lambda^{\star}} \frac{1}{(1+(n+1)G_{\beta,h}(0)) \lambda(p) + h}.
\end{equation}
The assumption of $\beta$ being finite-range and non-negative implies $\lambda(p) \leq C(\beta)|p|^2$.
If $d\leq 2$ it follows that
\begin{equation}
  \lim_{L\to\infty} \frac{1}{(2\pi L)^d} \sum_{p \in \Lambda^{\star}}
  \frac{1}{\lambda(p)+h} \uparrow \infty \quad \text{as } h\downarrow 0,
\end{equation}
and, as $G_{\beta,h}\geq 0$, this implies \eqref{e:GasympLE2}.
\end{proof}

\begin{remark} 
  \label{rk:localisation} 
  In fact, the proof shows $G_{\beta,h}(0) \geq c_\beta/\sqrt{\log h}$
  with $c_\beta >0$ when $h>0$ is small.  For the $\HH^{2|2}$ model,
  we conjecture that the optimal bound is
  $G_{\beta,h}(0) \asymp c_\beta/h$ for $h$ small, with $c_{\beta}>0$
  exponentially small as $\beta$ becomes large.  This is consistent
  with Anderson localisation.  On the other hand, for the $\HH^{n}$
  model with $n\geq 2$, localisation is not expected, i.e.,
  $G_{\beta,h}(0) \ll 1/h$.
\end{remark}

\subsubsection{Consequences for the vertex-reinforced jump process}
\label{sec:pf-rec}

In contrast to Corollary~\ref{cor:GasympLE2},
it has been proven \cite{MR2104878,MR2728731} that when $d\geq 3$
and $\beta_{ij}=\beta 1_{|i-j|=1}$,
\begin{equation} \label{e:GasympGE3}
  \lim_{h\downarrow 0} \lim_{L \to\infty} G_{\beta,h}(0) < \infty
\end{equation}
for all $\beta>0$ in the case of $\HH^2$ and for all sufficiently
large $\beta>0$ for $\HH^{2|2}$. In the $\HH^{2|2}$ case
\eqref{e:GasympGE3} corresponds to transience of the VRJP (in the
sense of bounded expected local time, see
Corollary~\ref{cor:vrjp-trans} below) and to the uniform boundedness
(in the spectral parameter $z\in\C_{+}$) of the expected square of the
absolute value of the resolvent for random band matrices in the sigma
model approximation \cite{MR2953867} (recall
Section~\ref{sec:intro}). It also implies that the hyperbolic symmetry
is spontaneously broken.

Due to the non-amenability of hyperbolic group actions, the question
of spontaneous symmetry breaking for hyperbolic sigma models is, in
general, subtle.  The usual formulations of the Mermin--Wagner theorem
for models with compact symmetries cannot hold in the non-amenable
case~\cite{MR2778578}, and, in fact, spontaneous symmetry breaking
appears to occur in all dimensions \cite{MR2153656,MR2189377}.
Nonetheless, \eqref{e:GasympLE2} and \eqref{e:GasympGE3} show that the
two-point function --- the observable of interest for the VRJP and the
random matrix problem --- does undergo a transition analogous to that
occurring in systems with compact symmetries.

\begin{proof}[Proof of Theorem~\ref{thm:recurrence}]
We must prove that for any translation invariant finite-range $\beta$
\begin{equation}
  \int_0^\infty \E^{\beta,\Z^d}_{0,0}(1_{X_t=0}) \, dt =\infty,
\end{equation}
where the expectation refers to that of the VRJP on $\Z^d$ and $d=1,2$.
This is true since, for any finite-range $\beta$, one has
\begin{align}
  \int_0^\infty \E^{\beta,\Z^d}_{0,0}(1_{X_t=0}) \, dt
  &=
    \lim_{h\downarrow 0} \int_0^\infty \E^{\beta,\Z^d}_{0,0}(1_{X_t=0})\, e^{-ht} \, dt
    \nonumber\\
  &=
    \lim_{h\downarrow 0}
    \lim_{L\to\infty} \int_0^\infty
    \E^{\beta,\Lambda_L}_{0,0}(1_{X_t=0}) \, e^{-ht} \, dt
    = \infty
    .
\end{align}
The first equality is by monotone convergence, and the final
equality is obtained by combining \eqref{e:GasympLE2} for the
$\HH^{2|2}$ model and \eqref{e:twopoint-H22-VRJP}.

For the second equality it suffices, by using the tail of the
exponential $e^{-ht}$, to verify that the integrand converges for
$t\leq T$ for any bounded $T$.  Since the jump rate $1+L_t^i$ is
bounded by $1+T$, the walk is exponentially unlikely to take more
than $O(T^{3})$ jumps to new vertices up to time $T$. 
VRJPs on $\Lambda_L$ and $\Z^d$ can be coupled to be the same
until they exit a ball of radius less than $\frac12 L$,
an event which requires at least $L/R$ jumps to occur,
where $R$ is the radius of the finite-range step distribution.
This completes the proof.
\end{proof}

The analogue of Theorem~\ref{thm:recurrence} for the ERRW with
constant initial weights was established in
\cite{MR2561431,SabotZengPubl}, but not for the VRJP. Mermin--Wagner type
theorems have also been proven for the ERRW in one and two
dimensions~\cite{MR2561431,MR2529432}.
The techniques used deal
directly with ERRWs, and hence are rather different from those
employed in this paper.

Our relation between the two-point functions of the $\HH^{2|2}$ model and the VRJP
also yields a transience result.

\begin{corollary}
  \label{cor:vrjp-trans}
  The vertex-reinforced jump process $(X_{t})$ on $\Z^{d}$, $d\geq 3$,
  with initial rates $\beta_{ij}=\beta 1_{|i-j|=1}$ and $\beta$
  sufficiently large is transient, in the sense that the expected time
  $(X_{t})$ spends at the origin is finite.
\end{corollary}

\begin{proof}
  The argument mirrors the proof of Theorem~\ref{thm:recurrence}, using
  \eqref{e:GasympGE3} in place of \eqref{e:GasympLE2}.
\end{proof}

Transience in the sense of visiting the origin finitely often almost
surely when $\beta$ is sufficiently large was established
in~\cite[Corollary~4]{MR3420510}; this result also makes use of
\cite{MR2728731}.  As with recurrence, see the discussion following
the statement of Theorem~\ref{thm:recurrence}, there is in general no
relation between the two notions of transience.

\section{Supersymmetry and horospherical coordinates}
\label{sec:susy}

In this section we define horospherical coordinates for $\HH^n$ and
then define the supersymmetric $\HH^{2|2}$ model precisely.  We also
collect Ward identities and relations between derivatives that will be
used in the proofs of Theorems~\ref{thm:VRJP-Dynkin} and \ref{thm:MW}.

\subsection{Horospherical coordinates}
\label{sec:horo}

As observed in \cite{MR1134935,MR2104878}, the hyperbolic spaces
$\HH^n$ are naturally parametrised by horospherical coordinates that
are useful for the analysis of the corresponding sigma models.
For $\HH^{n}$, these are global coordinates $t\in \R$,
$\tilde s \in \R^{n-1}$, in terms of which
\begin{equation} \label{e:xyzhoro}
  x = \sinh t -\frac12 |\tilde s|^{2} e^t,
  \;\; 
  y^i = e^{t} s^i 
  \;\;(i=1, \dots, n-1),
  \;\;
  z = \cosh t + \frac12 |\tilde s|^2 e^t.
\end{equation}
Both $x,z$ are scalars while
$\tilde y = (y^{1},\dots, y^{n-1})$ and
$\tilde s= (s^{1},\dots, s^{n-1})\in \R^{n-1}$ are $n-1$ dimensional vectors and
$|\tilde s|^2 = \sum_{i=1}^{n-1}(s^{i})^{2}$.
By this change of variables one has
(see Appendix~\ref{app:horo}),
\begin{equation}
  \label{e:HS-vol}
  \int_{(\HH^{n})^\Lambda} F(u) \,\mu^{\otimes\Lambda}(du)
  = \int_{(\R^{n})^{\Lambda}} F(u(\tilde s,t)) \, \prod_i
  e^{(n-1)t_i}\, dt_i \, d\tilde s_i.
\end{equation}
By a short calculation,
\begin{equation}
  \label{e:uxuy-notSUSY}
  -u_i \cdot u_j
  = \cosh(t_i-t_j) +\frac12 |\tilde s_i-\tilde s_j|^{2}e^{t_i+t_j},\qquad 
  z_i = \cosh t_i + \frac12 |\tilde s_i|^{2}e^{t_i}.
\end{equation}
Thus in horospherical coordinates,
\begin{multline} \label{e:Hhoro-Hn}
  H(\tilde s,t)
  = \sum_{ \avg{ij}} \beta_{ij} \pa{ \cosh(t_i-t_j)-1 +\frac12 |\tilde s_i-\tilde s_j|^{2}e^{t_i+t_j}}
  \\
  + h \sum_i \pa{ \cosh(t_i)-1 + \frac12 |\tilde s_i|^{2}e^{t_i}},
\end{multline}
where by a slight abuse of notation we have re-used the symbol
$H$. Moreover, the following relations, in which we set
$s_{i}=s_{i}^{r}$ and $y_{i}=y_{i}^{r}$ for some fixed
$r=1,\dots,n-1$, hold:
\begin{equation} \label{e:sderiv}
 \ddp{z_i}{s_i} = y_i, \qquad
  \ddp{y_i}{s_i} = x_i+z_i, \qquad
  \ddp{(u_i\cdot u_j)}{s_i} = y_j (x_i+z_i) - y_i(x_j+z_j).
\end{equation}
Furthermore,
\begin{equation} \label{e:d2zu}
  \begin{split}
\ddp{^2}{s_j^2} z_j
  &= e^{t_j} = x_j+z_j, \\
  \ddp{^2}{s_i\partial s_l} (-1-u_j \cdot u_l)
  &=
  \begin{cases}
   -e^{t_j+t_l} = -(x_j+z_j)(x_l+z_l), & i=j ,\\ 
   + e^{t_j+t_{l}} = + (x_j+z_j)(x_{l}+z_{l}), & i=l, \\ 
   0, & \text{else}.
  \end{cases}
  \end{split}
\end{equation}

\subsection{Supersymmetry}

Let $\Lambda$ be a finite set. We will define an algebra
$\Omega_\Lambda$ of forms (which generalise random variables) that
constitute the observables on the super-space $(\R^{2|2})^\Lambda$.
The super-space itself only has meaning through this algebra of
observables.  We also define an integral associated to this algebra.
We then introduce the supersymmetry
generator and the localisation lemma. For a more detailed introduction
to the mathematics of supersymmetry, see, e.g.,
\cite{MR2525670,MR1143413,MR2728731}.  

\subsubsection{Supersymmetric integration}
\label{sec:supersymm-integr}
For each vertex $i \in \Lambda$, let $x_i,y_i$ be real variables and
$\xi_i,\eta_i$ be two Grassmann variables.  Thus by definition all of
the $x_i$ and $y_i$ commute with each other and with all of the
$\xi_i$ and $\eta_i$ and all of the $\xi_i$ and $\eta_i$ anticommute.
The way in which the anticommutation relations are realized is
unimportant, but concretely, we can define an algebra of
$4^{|\Lambda|}\times 4^{|\Lambda|}$ matrices $\xi_i$ and $\eta_i$
realising the required anticommutation relations for the Grassmann
variables.
To fix signs in forthcoming expressions, fix an arbitrary order
$i_{1},\dots, i_{|\Lambda|}$ of the vertices in $\Lambda$.

We define the algebra $\Omega_\Lambda$ to be the algebra of smooth
functions on $(\R^2)^{\Lambda}$ with values in the algebra of
$4^{|\Lambda|}\times 4^{|\Lambda|}$ matrices that have the form
\begin{equation} \label{e:Fexpand}
  F = \sum_{I,J \subset \Lambda} F_{I,J}(x,y) (\eta\xi)_{I,J},
\end{equation}
where the coefficients $F_{I,J}$ are smooth functions on
$(\R^2)^{\Lambda}$, and $(\eta\xi)_{I,J}$ is given by the ordered
product $\prod_{i\in I\cap J}\eta_{i}\xi_{i}\prod_{i\in I\setminus
  J}\xi_{i}\prod_{j\in J\setminus I}\eta_{j}$. This ordering has been
chosen so that $(\eta\xi)_{\Lambda,\Lambda}$ is
$\eta_{1}\xi_{1}\dots\eta_{\Lambda}\xi_{\Lambda}$.  We call elements
of $\Omega_\Lambda$ forms because the forms of differential geometry
are instances \cite{MR2525670,MR941982}.  The integral (sometimes
called a superintegral) of a form $F\in \Omega_\Lambda$ is defined by
\begin{equation} \label{e:intFdef}
  \int_{(\R^{2|2})^\Lambda} F
  \bydef
  \int_{(\R^2)^{\Lambda}}
  F_{\Lambda,\Lambda}(x,y) \prod_{i\in\Lambda} \frac{dx_i \, dy_i}{2\pi},
\end{equation}
where $\R^{2|2}$ refers to the number of commuting and anticommuting variables.
 
The \emph{degree} of a coefficient $F_{I,J}$ is $|I|+|J|$.  Thus the
integral of a form $F$ is a constant multiple of the usual Lebesgue
integral of the top degree part of $F$.  A form $F \in \Omega_\Lambda$
is \emph{even} if the degree of all non-vanishing coefficients
$F_{I,J}$ is even in \eqref{e:Fexpand}.  Even forms commute.  For even
forms $F^1, \dots, F^p$ and a smooth function $g \in C^\infty(\R^p)$,
the form $g(F^1, \dots, F^p) \in \Omega_\Lambda$ is defined by
formally Taylor expanding $g$ about the degree-$0$ part
$(F^1_{\varnothing,\varnothing}(x,y), \dots,
F^p_{\varnothing,\varnothing}(x,y))$. This is well-defined as there is
no ambiguity in the ordering if the $F^{i}$ are all even, and the
anticommutation relations satisfied by the $\xi_{i}$ and $\eta_{i}$
imply the expansion is finite.

\subsubsection{Localisation}
\label{sec:localisation}

Temporarily set $x=x_{i}, y=y_{i},\xi=\xi_{i}$, and $\eta=\eta_{i}$.
Define an operator $\partial_{\eta} \colon \Omega_\Lambda \to \Omega_\Lambda$ by
linearity, $\partial_{\eta} (\eta F) = F$, and
$\partial_{\eta} F=0$ if $F$ does not contain a factor $\eta$.
Define $\partial_{\xi}$ in the same manner.
Define $Q_{i}$ by its action on forms $F$ by
\begin{equation}
  Q_{i}F \bydef \xi\partial_{x}F + \eta\partial_{y}F + x\partial_{\eta} F - y\partial_{\xi} F.
\end{equation}
The \emph{supersymmetry generator} $Q$ acts on a form $F \in \Omega_\Lambda$ by
$QF \bydef \sum_{i\in\Lambda} Q_iF$.

\begin{definition}
  $F \in \Omega_\Lambda$ is \emph{supersymmetric} if $QF=0$.
\end{definition}

The supersymmetry generator acts as an anti-derivation on the algebra of
forms, see, e.g., \cite[Section 6]{MR2525670}.  This implies that the forms
\begin{equation}
  \tau_{ji}= \tau_{ij}
  \bydef x_ix_j+y_iy_j + \xi_i\eta_j - \eta_i\xi_j, \qquad i,j \in \Lambda,
\end{equation}
are supersymmetric. Moreover, any smooth function of the $\tau_{ij}$ is
  supersymmetric as $Q$ obeys a chain rule, see~\cite[Equation~(6.5)]{MR2525670}.
The following \emph{localisation lemma} is fundamental. For a proof, see \cite[Lemma~16]{MR2728731}.

\begin{lemma}[Localisation lemma] \label{lem:SUSY-localisation} Let
  $F \in \Omega_\Lambda$ be a smooth form with sufficient decay that
  is supersymmetric, i.e., satisfies $QF=0$.  Then
  \begin{equation} \label{e:SUSY-localisation}
    \int_{(\R^{2|2})^\Lambda} F = F_{\varnothing,\varnothing}(0,0).
  \end{equation}
\end{lemma}

\subsection{The $\HH^{2|2}$ model}

We can now define the $\HH^{2|2}$ sigma model and justify our earlier
claim that its $y$ marginal is the probability measure
\eqref{e:H22-Real}.  Given $(x_{i},y_{i},\xi_{i},\eta_{i})$ as above
define an even variable $z_i$ by
\begin{equation}
  z_i
  \bydef \sqrt{1+x_i^2+y_i^2+2\xi_i\eta_i}
  = \sqrt{1+x_i^2+y_i^2} + \frac{\xi_i\eta_i}{\sqrt{1+x_i^2+y_i^2}},
\end{equation}
where the equality is by the definition of a function of a form.
We will write $u_i = (x_i,y_i,z_i,\xi_i,\eta_i)$. Define the ``inner product''
\begin{equation} \label{e:innerproduct-H22}
  u_i \cdot u_j \bydef x_ix_j+y_iy_j-z_iz_j+\xi_i\eta_j-\eta_i\xi_j,
\end{equation}
generalising the Minkowski inner product above \eqref{e:Hn-def}; we have written
``inner product'' as this is only terminology, since
\eqref{e:innerproduct-H22} is not a quadratic form in the classical
sense.  Then by a short calculation
\begin{equation}
  u_i \cdot u_i = -1,
\end{equation}
which we interpret as meaning that $u_i$ is in the supermanifold $\HH^{2|2}$.
Since $z_i = \sqrt{1+\tau_{ii}}$ and $u_i\cdot u_j = \tau_{ij} - z_iz_j$,
the forms $u_i \cdot u_j$ and $z_i$ are supersymmetric for all $i,j \in \Lambda$.

The $\HH^{2|2}$ integral of a form $F \in \Omega_\Lambda$ is defined by
\begin{equation}
  \label{e:integral-H22}
  \int_{(\HH^{2|2})^\Lambda} F
  \bydef \int_{(\R^{2|2})^\Lambda} F \prod_{i\in\Lambda} \frac{1}{z_i},
\end{equation}
and the $\HH^{2|2}$ model is defined by the following action (which is
now a form in $\Omega_\Lambda$)
\begin{equation} \label{e:H22-def}
  H\bydef H_{\beta,h} = \sum_{ \avg{ij}} \beta_{ij} (-u_i\cdot u_j-1) +
  h \sum_i (z_i-1) \in \Omega_\Lambda. 
\end{equation}
Lastly, we define the super-expectation of an observable
$F\in\Omega_\Lambda$ in the $\HH^{2|2}$ model by
\begin{equation} \label{e:H22-exp}
\avg{F}_{\HH^{2|2}} \bydef \int_{(\HH^{2|2})^{\Lambda}} F e^{-H}.
\end{equation}
Lemma~\ref{lem:SUSY-localisation} implies
that $\avg{1}_{\HH^{2|2}} = 1$, as promised in Section~\ref{sec:susy-hyp}.

\subsection{Supersymmetric horospherical coordinates}

The $\HH^{2|2}$ model can also be reparametrised in a supersymmetric
version of horospherical coordinates \cite[Sec.~2.2]{MR2728731}.
For the convenience of the reader, the explicit change of variables is
computed in Appendix~\ref{app:horo}.
In this parametrisation, $t$ and $s$ are two real variables and $\bar\psi$ and
$\psi$ are two Grassmann variables.  As in the previous
section, we denote the algebra of such forms by $\Omegahoro_\Lambda$.
The tilde refers to horospherical coordinates.  We write
\begin{equation} \label{e:xyzhorosusy}
  \begin{split}
  x = \sinh t -e^t (\frac12 s^2 +&\bar\psi\psi),
  \quad
  y = e^t s,
  \quad
  z = \cosh t + e^t(\frac12 s^2 +\bar\psi\psi), \\
  &\xi = e^t\bar\psi , \quad \eta = e^t \psi.
  \end{split}
\end{equation}
There is a generalisation of the change of variables formula from standard
integration to superintegration.
We only require the following special case given in \cite[Sec.~2.2]{MR2728731} and Appendix~\ref{app:horo}.
Forms $F \in \Omega_\Lambda$ are in correspondence with forms $\Fhoro \in \Omegahoro_\Lambda$
obtained by substituting the relations \eqref{e:xyzhorosusy} into \eqref{e:Fexpand}
using the definition of functions of forms.
Moreover, expanding
\begin{equation}
  \Fhoro
  =
  \sum_{I,J \subset \Lambda} \Fhoro_{I,J}(t,s)
  (\psi\bar\psi)_{I,J}
\end{equation}
the superintegral over $F$ can expressed as
\begin{equation}
  \label{e:HS-vol-SUSY}
  \int_{(\HH^{2|2})^\Lambda} F
  =
  \int_{(\R^2)^{\Lambda}} \Fhoro_{\Lambda,\Lambda}(t,s) \, \prod_i
  e^{-t_i}\, \frac{dt_i \, ds_i}{2\pi}.
\end{equation}

If a function $F(y)$ depends only on the $y$
coordinates then $F$ has degree $0$, and a computation
(see~\cite[Sec.~2.2]{MR2728731} and Appendix~\ref{app:horo}) shows that
\begin{align}
  \avg{F(y)}_{\HH^{2|2}} = \int_{(\HH^{2|2})^\Lambda} F(y) e^{-H}
  &= \int_{(\R^2)^\Lambda} F(e^{t}s) (e^{-H})_{\Lambda,\Lambda} \prod_{i} e^{-t_i} \frac{dt_i \, ds_i}{2\pi}
  \nonumber\\
  &= \int_{(\R^2)^\Lambda} F(e^{t}s) e^{-\Hhoro(t,s)} \prod_{i} dt_i \, ds_i,
\end{align}
with the function $\Hhoro$ given by \eqref{e:H22-Real}.

Analogously to \eqref{e:uxuy-notSUSY} a calculation gives the expressions
\begin{align}
  \label{e:uxuy-SUSY}
  -u_i \cdot u_j
  &= \cosh(t_i-t_j) +\frac12 (s_i-s_j)^2 e^{t_i+t_j} + (\bar\psi_i-\bar\psi_j)(\psi_i-\psi_j)e^{t_i+t_j}
  \\
  \label{e:zx-SUSY}
  z_i &= \cosh t_i + (\frac12 s_i^2 +\bar\psi_i\psi_i)e^{t_i}.
\end{align}
We again check that
\begin{equation} \label{e:sderiv-Susy}
  \ddp{z_i}{s_i} = y_i, \qquad
  \ddp{y_i}{s_i} = x_i+z_i, \qquad
  \ddp{(u_i\cdot u_j)}{s_i} = y_j (x_i+z_i) - y_i(x_j+z_j)
\end{equation}
and 
\begin{equation} \label{e:d2zu-susy}
  \begin{split}
  \ddp{^2}{s_j^2} z_j
  &= e^{t_j} = x_j+z_j,\\
  \ddp{^2}{s_i\partial s_l} (-1-u_j \cdot u_l)
  &=
  \begin{cases}
   -e^{t_j+t_l} = -(x_j+z_j)(x_l+z_l), & i=j ,\\ 
   + e^{t_j+t_{l}} = + (x_j+z_j)(x_{l}+z_{l}), & i=l, \\ 
   0, & \text{else}.
  \end{cases}
  \end{split}
\end{equation}

\subsection{Ward identities}
\label{sec:ward}

In this section we establish some useful Ward identities.
These Ward identities
are a reflection of the underlying symmetries of the target spaces
$\HH^{n}$ and $\HH^{2|2}$, see \cite[Appendix~B]{MR2728731}.  Note
that these identities are most easily seen in the ambient coordinates
$(x,y^1, \dots, y^{n-1},z)$.

\subsubsection{$\HH^n$}
\label{sec:hhn}

For the $\HH^{n}$ model we have the identities
\begin{equation}
  \label{e:nonSUSY-ward}
  \avg{x_{j}g(z)}_{\HH^{n}}=0.
\end{equation}
for any smooth function $g$.  This identity follows simply from the
invariance of the measure under $x \mapsto -x$
(see \eqref{e:HH-def}--\eqref{e:HH-exp}).
Moreover, by rotational symmetry,
we have $\avg{g(y^{r})}_{\HH^{n}}=\avg{g(x)}_{\HH^{n}}$ for
$r = 1,\dots, n-1$.

\subsubsection{$\HH^{2|2}$}
\label{sec:hh22}

For the $\HH^{2|2}$ model we have identities analogous to \eqref{e:nonSUSY-ward}:
\begin{equation}
  \label{e:SUSY-basicward}
  \avg{x_{j}g(z)}_{\HH^{2|2}}=0
\end{equation}
for any smooth function $g$. This identity again follows from the
symmetry $x \mapsto -x$ (see \eqref{e:H22-def}--\eqref{e:H22-exp}).
We also have $\avg{g(x)}_{\HH^{2|2}}=\avg{g(y)}_{\HH^{2|2}}$ by
rotational symmetry.  The following identities arise from
\eqref{e:SUSY-basicward}:
\begin{equation} 
  \label{e:SUSY-ward-pre}
  \begin{split}
  \avg{e^{t_j+t_l}}_{\HH^{2|2}} &= \avg{(x_j+z_j)(x_l+z_l)}_{\HH^{2|2}}
                                  =  \avg{x_jx_l+z_jz_l}_{\HH^{2|2}}
  \\ 
  \avg{e^{t_j}}_{\HH^{2|2}} &= \avg{x_j+z_j}_{\HH^{2|2}} %= 1.
  \end{split}
\end{equation}
and hence by supersymmetry and rotational invariance
\begin{equation} 
  \label{e:SUSY-ward}
  \begin{split}
  \avg{e^{t_j+t_l}}_{\HH^{2|2}} &= 1+\avg{y_jy_l}_{\HH^{2|2}}, 
  \\ 
  \avg{e^{t_j}}_{\HH^{2|2}} &= 1.
  \end{split}
\end{equation}
Indeed, the evaluations $\avg{z_iz_j}_{\HH^{2|2}} =
{\avg{z_{i}}_{\HH^{2|2}}}=1$ are by Lemma~\ref{lem:SUSY-localisation},
which implies more generally that for any smooth function $g$ with rapid decay,
\begin{equation} \label{e:localisation-z}
  \int_{(\HH^{2|2})^\Lambda} e^{-H_{\beta,0}} g(z) = g(1).
\end{equation}

\section{Proof of Theorem~\ref{thm:VRJP-Dynkin}}

In this section, for the $\HH^n$ model, we will let $y_{a}$ denote the
component $y^{1}_{a}$ of $u_{a}\in\HH^{n}$ and $s_{a}$ the
corresponding component $s^1_a$ in horospherical coordinates. By
symmetry (recall Section~\ref{sec:ward}), the results of this section
are valid if we replace $y^1_a$ by any of the first $n-1$ components
of $u_a$.

We will prove that for the $\HH^{n}$ model, $n\geq 2$,
\begin{equation} \label{e:Eg-Dynkin-VRJP}
  \begin{split}
  \hspace{-5mm}\sum_b
  \int_{(\HH^n)^\Lambda} e^{-H_{\beta,h}}& y_a y_b g(b,z-1) =\\
  &\int_{(\HH^{n})^\Lambda} e^{-H_{\beta,h}} z_a \int_0^\infty
  \E^{\beta}_{a,z-1}(g(X_t,L_t)) \, e^{-ht} \, dt .
  \end{split}
\end{equation}
In \eqref{e:Eg-Dynkin-VRJP}, and in the rest of this section, we omit the measure
$\mu^{\otimes\Lambda}(du)$ for integrals over $(\HH^n)^\Lambda$ from the notation.
For the $\HH^{2|2}$ model we prove that
\begin{equation} \label{e:Eg-VRJP}
  \sum_b \int_{(\HH^{2|2})^\Lambda}
  e^{-H_{\beta,h}} y_a y_b g(b,z-1) = \int_0^\infty
  \E^{\beta}_{a,0}(g(X_t,L_t)) \, e^{-ht} \, dt.
\end{equation} 
Theorem~\ref{thm:VRJP-Dynkin} in the case of $\HH^{2|2}$ is precisely
\eqref{e:Eg-VRJP}, and Theorem~\ref{thm:VRJP-Dynkin} in the case of
$\HH^{n}$ follows by normalising \eqref{e:Eg-Dynkin-VRJP}.  The
identities \eqref{e:Eg-Dynkin-VRJP} and \eqref{e:Eg-VRJP} are a result
of the following integration by parts formulas.  Recall that
$\mathcal{L}^\beta$ denotes the generator \eqref{e:Ldef} of the joint
position and local time process $(X_t,L_t)$ of the VRJP.

\begin{lemma} \label{thm:IBP}
  Let $\Lambda$ be finite, let $a\in\Lambda$, and let $g\colon \Lambda \times \R^\Lambda \to \R$ be a
  smooth function with rapid decay.  For the $\HH^{n}$ model, $n\geq 2$,
  \begin{equation} \label{e:Lg-Dynkin-VRJP}
    -  \sum_b \int_{(\HH^{n})^\Lambda} e^{-H_{\beta,0}} y_ay_b \mathcal{L}^{\beta}g(b,z-1)
    = \int_{(\HH^{n})^\Lambda} e^{-H_{\beta,0}} z_a g(a,z-1).
  \end{equation}
  For the $\HH^{2|2}$ model,
  \begin{equation} \label{e:Lg-VRJP}
    -  \sum_b \int_{(\HH^{2|2})^\Lambda} e^{-H_{\beta,0}} y_a y_b \mathcal{L}^{\beta}g(b,z-1)
    = g(a,0).
  \end{equation}
\end{lemma}
\begin{proof}
The proofs are essentially the same for $\HH^{n}$ and
$\HH^{2|2}$, so we carry them out in parallel.

We write $\mathcal{L}$ for $\mathcal{L}^{\beta}$, $H$ for
$H_{\beta,0}$, and the integral $\int$ 
for $\int_{(\HH^n)^\Lambda}$ and, respectively, $\int_{(\HH^{2|2})^\Lambda}$.
By \eqref{e:sderiv} (resp.\ \eqref{e:sderiv-Susy}) we have
$y_b \ddp{}{\ell_b} g(b,z-1) = \ddp{}{s_b} g(b,z-1)$ where
$\ddp{}{\ell_b}$ denotes the derivative with respect to the $b$-th
component of the second argument.  Therefore
\begin{multline} \label{e:IBPstart}
  \sum_b \int e^{-H} y_a y_b \mathcal{L}g(b,z-1)
  \\
  = \int e^{-H} y_a \pbb{ \sum_{b,c} \beta_{bc} y_b z_c
    (g(c,z-1)-g(b,z-1)) + \sum_b \ddp{}{s_b} g(b,z-1) }  .
\end{multline}
Recall \eqref{e:HS-vol} (resp.\ \eqref{e:HS-vol-SUSY}) and
integrate the second term in the equation above by parts.
This produces two terms; by the rapid decay of $g$ there are no boundary terms.
For the first term produced by the integration by parts,
using \eqref{e:sderiv} (resp.\ \eqref{e:sderiv-Susy}) again,
\begin{align}
  \nonumber
  &\hspace{-15mm}\sum_b \int e^{-H} y_a \pa{-\ddp{H}{s_b}} g(b,z-1)\\
  \nonumber
  &=
  \sum_b \int e^{-H} y_a \pa{\sum_c \beta_{bc} \ddp{(u_b\cdot u_c)}{s_b}} g(b,z-1)
  \\
  &=
  \sum_{b,c} \int e^{-H} y_a \beta_{bc} y_bz_c (g(c,z-1)-g(b,z-1))
  .
\end{align}
This term cancels the first term on the right-hand side of \eqref{e:IBPstart}.
For the second term produced by the integration by parts, we use that
$\int x_a e^{-H} g(b,z) = 0$ by
\eqref{e:nonSUSY-ward} (resp.\ \eqref{e:SUSY-basicward}):
\begin{align} \label{e:Lg-VRJP-Dykin-pf-localisation} \nonumber
  \int e^{-H} \ddp{y_a}{s_b} g(b,z-1)
  &=
  \delta_{ab} \int e^{-H} (x_a+z_a) g(b,z-1)\\
  &=
  \delta_{ab} \int e^{-H} z_a g(a,z-1).
\end{align}
In the supersymmetric case, the localisation lemma
in the special case \eqref{e:localisation-z} further implies
that the last right-hand side can be evaluated as
\begin{equation} \label{e:Lg-VRJP-pf-localisation}
  \delta_{ab} \int e^{-H} z_a g(a,z-1)
  =
  \delta_{ab} g(a,0).
\end{equation}
Altogether, we have shown \eqref{e:Lg-Dynkin-VRJP} (resp.\ \eqref{e:Lg-VRJP}).
\end{proof}

\begin{proof}[Proof of Theorem~\ref{thm:VRJP-Dynkin}]
  It suffices to show \eqref{e:Eg-Dynkin-VRJP} and \eqref{e:Eg-VRJP}
  with $h=0$, by replacing $g(b,z-1)$ by $g(b,z-1)e^{-h(z-1)}$.
  Therefore from now on assume $h=0$.  To get \eqref{e:Eg-VRJP} from
  \eqref{e:Lg-VRJP}, we apply \eqref{e:Lg-VRJP} with $g(i,\ell)$
  replaced by $g_t(i,\ell) = \E_{i,\ell}(g(X_t,L_t))$.  By the
  definition of the generator we have
  $\mathcal{L} g_t(i,\ell) = \ddp{}{t} g_t(i,\ell)$, so
  \eqref{e:Lg-VRJP} gives
\begin{equation}
  \E_{a,0}(g(X_t,L_t)) =
  -\ddp{}{t} \pa{
  \sum_b \int e^{-H} y_a y_b g_t(b,z-1)}.
\end{equation}
Note that the process $(X_t,L_t)$ is transient even if the marginal $(X_t)$
is recurrent because $\sum_i L^i_t \to\infty$ as $t\to\infty$.
Therefore, integrating both sides over $t$ and using that $g_t(x,\ell) \to 0$ as
$t\to\infty$, which follows from the transience of $(X_{t},L_t)$ and the rapid decay
of $g=g_{0}$, we get
\begin{equation}
  \int_0^\infty
  \E_{a,0}(g(X_t,L_t)) \, dt
  =
  \sum_b \int e^{-H} y_a y_b g(b,z-1).
\end{equation}
The proof of \eqref{e:Eg-Dynkin-VRJP} from \eqref{e:Lg-Dynkin-VRJP} is
entirely analogous.
\end{proof}

\section{Proof of Theorem~\ref{thm:MW}}

The proof of the hyperbolic Mermin--Wagner follows that of the usual
Mermin--Wagner theorem closely
\cite{PhysRevLett.17.1133,10.1063/1.1705316}; see also the
presentation in \cite{MR610687}.  We begin with the non-supersymmetric
case.  Due to the non-compact target space, differences occur in the
bound of the term $\avg{|DH|^2}$ and in the role of the coordinate in
the direction of the magnetic field. As in the previous section we
write $H$ for $H_{\beta,h}$. We will write $\bar A$ to denote the
complex conjugate of $A$.

\begin{proof}[Proof of \eqref{e:MW-Hn}]
As in the previous section we write $y_{j}$ for $y_{j}^{1}$. We
also write $\avg{\cdot}$ for $\avg{\cdot}_{\HH^{n}}$,
and we use horospherical coordinates throughout the proof.
Throughout the proof $H$ will denote the energy of a spin
configuration in horospherical coordinates, recall \eqref{e:Hhoro-Hn}.

Let
\begin{equation}
  S(p) = \frac{1}{\sqrt{|\Lambda|}} \sum_j e^{i(p\cdot j)} y_j,
  \quad
  D = \frac{1}{\sqrt{|\Lambda|}} \sum_j e^{-i(p\cdot j)} \ddp{}{s_j}.
\end{equation}
By the Cauchy--Schwarz inequality,
\begin{equation} \label{e:CS}
  \avg{|S(p)|^2} \geq \frac{|\avg{S(p)DH}|^2}{\avg{|DH|^2}}.
\end{equation}

In the following, we compute the terms on the left- and right-hand
sides of the above inequality.
Note that we have the integration by
parts identity $\avg{F DH} = \avg{DF}$ for any smooth
$F\colon (\HH^{n})^{\Lambda}\to \R$ that does not grow too fast; the vanishing of
boundary terms can be seen by looking at the expression for $H$
(i.e., by \eqref{e:Hhoro-Hn}).

By the assumed translation invariance of $\beta$,
\begin{align}
  \label{e:S2}
  \avg{|S(p)|^2}
  &= \frac{1}{|\Lambda|} \sum_{j,l} e^{i p\cdot(j-l)} \avg{y_jy_l}
  = \frac{1}{|\Lambda|} \sum_{j,l} e^{i p\cdot(j-l)} \avg{y_0y_{j-l}}
  \\
  \nonumber
  &= \sum_{j} e^{i (p\cdot j)} \avg{y_0y_{j}},
 \\
  \label{e:SH}
  \avg{S(p)DH}
  &= \avg{DS(p)}
  = \frac{1}{|\Lambda|} \sum_{j,l} e^{ip\cdot(j-l)}\avg{\ddp{y_j}{s_l}}
  = \frac{1}{|\Lambda|} \sum_{j} \avg{x_j+z_j} \\
  \nonumber
  &= \avg{z_0},
\\
  \label{e:H2}
  \avg{|DH|^2}
  &
  = \avg{D\bar D H}
    = \frac{1}{|\Lambda|} \sum_{j,l} e^{ip \cdot (j-l)}\avga{\ddp{^2H}{s_j \partial s_l}}
    .
\end{align}
In \eqref{e:SH} we have used $\avg{x_{j}}=0$; recall
Section~\ref{sec:ward}. By $\avg{x_jz_k} = 0$, Cauchy--Schwarz,
translation invariance,
that $\avg{x_0^2} = \avg{y_0^2}$ (recall the symmetries from Section~\ref{sec:hhn}),
and the constraint $u_0 \cdot u_0 = -1$, observe that
\begin{equation}
  \avg{(x_j+z_j)(x_l+z_l)}
  =
  \avg{x_jx_l+z_jz_l}
  \leq
  \avg{x_0^2}+\avg{z_0^2}
  =
  1+(n+1)\avg{y_0^2}
  .
\end{equation}
Thus, using \eqref{e:d2zu} and $\avg{x_j} = 0$ once more,
\eqref{e:H2} can be rewritten and bounded above by
\begin{align}
  \nonumber
  \avg{|DH|^2}
  &=
    \frac{1}{|\Lambda|} \sum_{j,l} \beta_{jl} \avg{(x_j+z_j)(x_l+z_l)}
    (1-e^{ip\cdot (j-l)}) + \frac{h}{|\Lambda|} \sum_j \avg{x_j+z_j} 
    \\
  &\leq
    \frac{1}{|\Lambda|} \sum_{j,l} \beta_{jl} (1+(n+1)\avg{y_0^2}) (1-\cos(p\cdot (j-l))) + h\avg{z_0}.
\end{align}
In summary, we have shown (recall \eqref{e:lambda})
\begin{align}
  \avg{|DH|^2}
  \leq
  (1+(n+1)\avg{y_0^2}) \lambda(p) + h\avg{z_0}.
\end{align}
Using \eqref{e:S2} and substituting the above bounds into \eqref{e:CS} gives
\begin{align}
  \nonumber
  \sum_{j} e^{i (p\cdot j)} \avg{y_0y_{j}}
  \geq \frac{|\avg{S(p)DH}|^2}{\avg{|DH|^2}}
  &\geq
    \frac{\avg{z_0}^2}{(1+(n+1)\avg{y_0^2})\lambda(p) + h\avg{z_0}}
    \\
  &\geq
    \frac{1}{ (1+(n+1)\avg{y_0^2}) \lambda(p) + h}.
\end{align}
The last inequality follows from $h\geq 0$ and $1\leq \avg{z_0}$,
which holds by the definition of $\HH^n$.
\end{proof}

\begin{proof}[Proof of \eqref{e:MW-H22}]
  We use that the expectation of a function $F(y)$ can be written
  using horospherical coordinates in terms of the \emph{probability
    measure} \eqref{e:H22-Real}.  Throughout this proof, we denote the
  expectation with respect to this probability measure by
  $\avg{\cdot}$.  By the Cauchy--Schwarz inequality, and since $S(p)$
  is a function of the $y$,
  \begin{equation} 
    \label{e:CS-SUSY} 
    \avg{|S(p)|^2}_{\HH^{2|2}} =
    \avg{|S(p)|^2} \geq
    \frac{|\avg{S(p)D\Hhoro}|^2}{\avg{|D\Hhoro|^2}}.
  \end{equation}
  The probability measure $\avg{\cdot}$ obeys the integration by parts
  $\avg{FD\Hhoro} = \avg{DF}$ identity for any function $F=F(s,t)$
  that does not grow too fast. Therefore by translation invariance we
  find that, as in the case of $\HH^{n}$, 
  \begin{align}
    \label{e:S2-susy}
    \avg{|S(p)|^2}
    &= \frac{1}{|\Lambda|} \sum_{j,l} e^{i p\cdot(j-l)} \avg{y_jy_l}
      = \frac{1}{|\Lambda|} \sum_{j,l} e^{i p\cdot(j-l)}
      \avg{y_0y_{j-l}} \\
    \nonumber
      &= \sum_{j} e^{i (p\cdot j)} \avg{y_0y_{j}},
    \\
    \label{e:SH-susy}
    \avg{S(p)D\Hhoro}
    &= \avg{DS(p)}
      = \frac{1}{|\Lambda|} \sum_{j,l} e^{ip\cdot(j-l)}\avg{\ddp{y_j}{s_l}}
      = \frac{1}{|\Lambda|} \sum_{j} \avg{e^{t_j}}
      = 1,
  \end{align}
  where the last identity uses \eqref{e:SUSY-ward}.  By
  \eqref{e:SUSY-ward}, Cauchy--Schwarz, and translation invariance we have
  \begin{equation}
    \label{e:af}
    \avg{e^{t_j+t_l}}
    = 1+\avg{y_jy_l}
    \leq 1+\avg{y_0^2} .
  \end{equation}
  Using \eqref{e:af} and the integration by parts identity it
  follows that 
  \begin{align}
    \avg{|D\Hhoro|^2}
    =
    \avg{D\bar D\Hhoro}
    &=
      \frac{1}{|\Lambda|} \sum_{j,l} \beta_{jl} \avg{e^{t_j+t_l}} (1-\cos(p\cdot (j-l))) + \frac{h}{|\Lambda|} \sum_j \avg{e^{t_j}}
          \nonumber\\
    &\leq
      \frac{1}{|\Lambda|} \sum_{j,l} \beta_{jl} (1+\avg{y_0^2}) (1-\cos(p\cdot (j-l))) + h
      \nonumber\\
    &=
      (1+\avg{y_0^2})  \lambda(p) + h.
  \end{align}
  In summary, we have proved
  \begin{equation}
    \sum_{j} e^{i (p\cdot j)} \avg{y_0y_{j}}
    =
    \avg{|S(p)|^2}
    \geq \frac{|\avg{S(p)D\Hhoro}|^2}{\avg{|D\Hhoro|^2}}
    \geq
    \frac{1}{(1+\avg{y_0^2}) \lambda(p) + h}
  \end{equation}
  as claimed.
\end{proof}

\appendix
\section{Horospherical coordinates}
\label{app:horo}

\subsection{$\HH^{n}$}
Under the change of variables 
\begin{equation}
x = \sinh t -\frac12 |\tilde s|^{2} e^t,
\quad
y^i = e^{t} s^i, \quad
z = \cosh t + \frac12 |\tilde s|^2 e^t,
\end{equation}
the measure transforms as
\begin{equation}
\frac{1}{z} \, dx \wedge dy^{1} \wedge \dots \wedge dy^{n-1}\mapsto \frac{\det J}{\cosh t + \frac12 |\tilde s|^2 e^t}dt\wedge ds^{1} \wedge \dots \wedge ds^{n-1},
\end{equation}
where the Jacobian matrix in block form is
\begin{equation}
J = \begin{bmatrix}
A_{1\times 1}&B_{1\times n-1}\\
C_{n-1\times 1}&D_{n-1\times n-1}
\end{bmatrix}
\end{equation}
with
\begin{alignat}{2}
A &= \ddp{x}{t} = \cosh{t} - \frac{1}{2}|\tilde{s}|^2e^t, &\qquad B_j &= \ddp{x}{s^j} = -s_je^{t},\\
C_i &= \ddp{y^i}{t} = s^{i}e^t, &\qquad D_{ij} &= \ddp{y^i}{s^j} = \delta_{ij}e^{t}.
\end{alignat}
Noting that $D = e^{t}I$, the determinant is easily computed using the
Schur complement formula,
\begin{align}
\det J &= (\det D)\det{(A-BD^{-1}C)}\nonumber\\
&= e^{(n-1)t}\left(\cosh{t} - \frac{1}{2}|\tilde{s}^2|e^{t} -
  \sum_{i=1}^{n-1}(-s_ie^{t})e^{-t}(s_ie^{t})\right) \nonumber \\
&= e^{(n-1)t}(\cosh{t} + \frac{1}{2}|\tilde{s}|^2e^{t}),
\end{align}
giving the transformed measure as
\begin{equation}
\frac{\det J}{\cosh t + \frac12 |\tilde s|^2 e^t} \, dt \wedge ds^{1} \wedge \dots \wedge ds^{n-1} = e^{(n-1)t} \, dt\wedge ds^{1} \wedge \dots \wedge ds^{n-1}.
\end{equation}

\subsection{$\HH^{2|2}$}
The calculation for $\HH^{2|2}$ is similar to the previous case, but
the Jacobian is replaced by the Berezinian.
The notation in \eqref{e:intFdef} 
corresponds to the following notation in 
\cite{MR2728731} resp.\ \cite{MR914369}:
\begin{equation}
  \int_{\R^{2|2}} F
  = \int dx \wedge dy \circ \partial_\xi \, \partial_\eta \, F
  = \int F \, d_\eta \, d_\xi \, dx \, dy .
\end{equation}
Applying \cite[Theorem~2.1]{MR914369} to the change of variables
\begin{equation}
  \begin{split}
  x = \sinh t -\frac12(s^{2}+&2\psibar\psi) e^t,\quad 
  y = se^{t}, \quad 
  z = \cosh t + \frac12 (s^2 + 2\psibar\psi)e^t, \\
  &
  \eta = \psi e^{t},\quad 
  \xi = \psibar e^{t},\quad
 \end{split}
\end{equation}
the Berezin measure transforms as
\begin{equation} 
  \frac{1}{z} \, d_\eta \, d_\xi \, dx \, dy \mapsto 
  \frac{\sdet M}{\cosh t + \frac12 (s^2 + 2\psibar\psi)e^t} \, d_\psi \, d_{\psibar} \, dt \, ds,
\end{equation}
where $M$ is the Berezinian supermatrix
\begin{equation}
M = \begin{bmatrix}
A&B\\
C&D
\end{bmatrix}
= 
\begin{bmatrix}
\ddp{x}{t}&\ddp{y}{t}&\ddp{\eta}{t}&\ddp{\xi}{t}\\
\ddp{x}{s}&\ddp{y}{s}&\ddp{\eta}{s}&\ddp{\xi}{s}\\
\ddp{x}{\psi}&\ddp{y}{\psi}&\ddp{\eta}{\psi}&\ddp{\xi}{\psi}\\
\ddp{x}{\psibar}&\ddp{y}{\psibar}&\ddp{\eta}{\psibar}&\ddp{\xi}{\psibar}\\
\end{bmatrix},
\end{equation}
and $\sdet M=(\det{D})^{-1}\det{(A -BD^{-1}C)}$ is its Berezinian (superdeterminant).	
The four blocks are then
\begin{alignat}{2}
A &= \begin{bmatrix}
\cosh{t}-\frac{1}{2}(s^2+2\psibar\psi)e^{t}&se^{t}\\
-se^{t}&e^{t}\\
\end{bmatrix},
&\qquad B &= \begin{bmatrix}
\psi e^{t}&\psibar e^{t}\\
0&0\\
\end{bmatrix},\\
C&= \begin{bmatrix}
\psibar e^{t}&0\\
-\psi e^{t}&0
\end{bmatrix},
&\qquad
D &= \begin{bmatrix}
e^{t}&0\\
0&e^{t}
\end{bmatrix}.
\end{alignat}
The first term in the Berezinian is simply
$(\det{D})^{-1} = e^{-2t}$,
whilst the second is
\begin{align}
  \nonumber
  \det{(A -BD^{-1}C)} 
&= \det{\left(\begin{bmatrix}
	\cosh{t}-\frac{1}{2}(s^2+2\psibar\psi)e^{t}&se^{t}\\
	-se^{t}&e^{t}\\
	\end{bmatrix}
	+
	\begin{bmatrix}
	2\psibar\psi e^{t}&0\\
	0&0
	\end{bmatrix} 
	\right)}\\
&= e^{t}\left(\cosh{t} + \frac{1}{2}(s^2 + 2 \psibar\psi)e^{t}\right),
\end{align}
giving the transformed Berezin measure as
\begin{equation}
\frac{\sdet M}{\cosh t + \frac12 (s^2 + 2\psibar\psi)e^t} \, d_\psi \, d_{\psibar} \, dt \, ds = e^{-t} \, d_\psi \, d_{\psibar} \,  dt \, ds,
\end{equation}
which corresponds to \eqref{e:HS-vol-SUSY}.

\section*{Acknowledgments}

We would like to thank D.\ Brydges, M.\ Disertori, M.\ Niedermaier,
T.\ Spencer, and B.\ T\'{o}th for helpful discussions on the topics of
this paper. We would also like to thank the referee for their
constructive comments and suggestions.  T.H.\ is supported by EPSRC
grant no.\ EP/P003656/1.  A.S.\ is supported by EPSRC grant no.\
1648831.

%\bibliography{all}
%\bibliographystyle{plain}

\end{document}